\documentclass[a4paper]{amsart} 

\usepackage{amsmath}
\usepackage{amssymb}
\usepackage{amsfonts}
\usepackage {graphicx}
\usepackage[latin1]{inputenc}
\usepackage[ruled,vlined,titlenumbered,linesnumbered,longend]{algorithm2e}

\usepackage{tikz}

\newtheorem{theorem}{Theorem}[section]
\newtheorem{lemma}[theorem]{Lemma}

\newtheorem{proposition}[theorem]{Proposition}

\newtheorem{definition}[theorem]{Definition}
\newtheorem{remark}[theorem]{Remark}
\newtheorem{example}[theorem]{Example}

\newcommand{\B}[1]{\fbox{#1}\;}
\newcommand{\pg}{\left\{ }
\newcommand{\pd}{\right\} }

\begin{document}
\title{Boltzmann Samplers for Colored Combinatorial Objects}
\author{O. Bodini, A. Jacquot}
\maketitle
\begin{center}
\address{
LIP6, UMR 7606, Departement CALSCI,
Universit\'e Paris 6 - UPMC,\\
104, avenue du pr\'esident Kennedy,\\ 
F-75252 Paris cedex 05, France\\
}

\email{\{Olivier.Bodini, Alice.Jacquot\}@lip6.fr}
\end{center}
\begin{abstract}
In this paper, we give a general framework for the Boltzmann generation of colored objects belonging to combinatorial constructible classes. We propose an intuitive notion called \textit{profiled objects} which allows the sampling of size-colored objects (and also of $k$-colored objects) although the corresponding class cannot be described by an analytic ordinary generating function. 
\end{abstract}


Efficient generation of extremely large objects is needed in many situations.  For instance in statistical physics for observing limit behaviors \cite{BFP},  in biology  for understanding  and analysing genome properties \cite{PTD}, or in computer sciences for testing programs \cite{DGGLP} or  simulating and modelizing Network as Internet \cite{DS,BDS}. In 2003, Duchon, Flajolet, Louchard and Schaeffer \cite{DFLS}  proposed a new model, called Boltzmann model, which leads to  systematically construct samplers for random generation of objects in combinatorial constructible classes. 
These samplers depends on
 a real parameter $x$ and generate an object $a$ in a combinatorial constructible class $\mathcal{A}$ with a probability essentially proportional to $x^{|a|}$ where $|a|$ is the size of $a$. Hence they draw uniformely in the class $\mathcal{A}_n$ of all the objects of size $n$ in $\mathcal{A}$. The size of the output is a random variable, and parameter $x$ can be tuned for a targetted  mean value. Moreover using rejection, one can obtain exact size samplers or approximate size samplers. This new approach differs from the 
 ''recursive method'' introduced by Nijenhuis and Wilf \cite{NW} by giving the possibility of relaxing the constraint of an exact size for the output and this implies a significant gain in complexity: no preprocessing phase is needed and expected time complexity is linear in the size of the output.
 
Boltzmann model has been described both in the cases of unlabelled combinatorial constructible (also called specifiable or decomposable) classes \cite{DFLS} and labelled combinatorial constructible classes \cite{DFLS,FFP}, for the most classical constructions (+, $\times$, Seq, Set, Cyc,...). 
In this paper, we are interested in the generation under Boltzmann model of size-colored combinatorial objects. Let $a$ be a combinatorial object with $n$ atoms, we say that this object is \emph{size-colored or colored} if each of its atoms can be colored with a color in $\{1,...,n\}$. Our main motivation for this study stems from the following situation : Consider that the construction of an object of size $n$ is distributed to $n$ heterogenous processors. Each processors signs the atoms that it has build. A size-colored object is exactly a possible way to build such an object and we would like to highlight some properties about it.  
In order to do that, we are going to deal with well-known $k$-colored combinatorial object (in this case, each atom can be colored with a color in $\{1,...,k\}$, and $k$ does not depend on the size). $k$-colored classes can be found in numerous problems as the $k$-colored necklaces \cite{FM}, expression trees with $k$ types of $n$-ary functions, $k$-colored planar trees, $k$-colored Motzkin paths \cite{ST},...

\newpage

For instance, all the 2-colored non-planar general trees of size 3 are :

{\small
\begin{tikzpicture}[baseline=(2),level distance=5mm,sibling distance=4mm]
\node (1) {\scriptsize{1}}
    child {
	node (2) {\scriptsize{1}}
	child {
	node {\scriptsize{1}}}
    }
;
\end{tikzpicture}
\begin{tikzpicture}[baseline=(2),level distance=5mm,sibling distance=4mm]
\node (1) {\scriptsize{1}}
    child {
	node (2) {\scriptsize{1}}
	child {
	node {\scriptsize{2}}}
    }
;
\end{tikzpicture}
\begin{tikzpicture}[baseline=(2),level distance=5mm,sibling distance=4mm]
\node (1) {\scriptsize{1}}
    child {
	node (2) {\scriptsize{2}}
	child {
	node {\scriptsize{1}}}
    }
;
\end{tikzpicture}
\begin{tikzpicture}[baseline=(2),level distance=5mm,sibling distance=4mm]
\node (1) {\scriptsize{1}}
    child {
	node (2) {\scriptsize{2}}
	child {
	node {\scriptsize{2}}}
    }
;
\end{tikzpicture}
\begin{tikzpicture}[baseline=(2),level distance=5mm,sibling distance=4mm]
\node (1) {\scriptsize{2}}
    child {
	node (2) {\scriptsize{1}}
	child {
	node {\scriptsize{1}}}
    }
;
\end{tikzpicture}
\begin{tikzpicture}[baseline=(2),level distance=5mm,sibling distance=4mm]
\node (1) {\scriptsize{2}}
    child {
	node (2) {\scriptsize{1}}
	child {
	node {\scriptsize{2}}}
    }
;
\end{tikzpicture}
\begin{tikzpicture}[baseline=(2),level distance=5mm,sibling distance=4mm]
\node (1) {\scriptsize{2}}
    child {
	node (2) {\scriptsize{2}}
	child {
	node {\scriptsize{1}}}
    }
;
\end{tikzpicture}
\begin{tikzpicture}[baseline=(2),level distance=5mm,sibling distance=4mm]
\node (1) {\scriptsize{2}}
    child {
	node (2) {\scriptsize{2}}
	child {
	node {\scriptsize{2}}}
    }
;
\end{tikzpicture}
\begin{tikzpicture}[baseline=(2),level distance=5mm,sibling distance=4mm]
\node (1) {\scriptsize{1}}
    child {
	node {\scriptsize{1}}
}
	child {
	node {\scriptsize{1}}}
    
;
\end{tikzpicture}
\begin{tikzpicture}[baseline=(2),level distance=5mm,sibling distance=4mm]
\node (1) {\scriptsize{1}}
    child {
	node {\scriptsize{1}}
}
	child {
	node {\scriptsize{2}}}
    
;
\end{tikzpicture}
($\equiv$
\begin{tikzpicture}[baseline=(2),level distance=5mm,sibling distance=4mm]
\node (1) {\scriptsize{1}}
    child {
	node {\scriptsize{2}}
}
	child {
	node {\scriptsize{1}}}
    
;
\end{tikzpicture}
)
\begin{tikzpicture}[baseline=(2),level distance=5mm,sibling distance=4mm]
\node (1) {\scriptsize{1}}
    child {
	node {\scriptsize{2}}
}
	child {
	node {\scriptsize{2}}}
    
;
\end{tikzpicture}
\begin{tikzpicture}[baseline=(2),level distance=5mm,sibling distance=4mm]
\node (1) {\scriptsize{2}}
    child {
	node {\scriptsize{1}}
}
	child {
	node {\scriptsize{1}}}
    
;
\end{tikzpicture}
\begin{tikzpicture}[baseline=(2),level distance=5mm,sibling distance=4mm]
\node (1) {\scriptsize{2}}
    child {
	node {\scriptsize{1}}
}
	child {
	node {\scriptsize{2}}}
    
;
\end{tikzpicture}
($\equiv$
\begin{tikzpicture}[baseline=(2),level distance=5mm,sibling distance=4mm]
\node (1) {\scriptsize{2}}
    child {
	node {\scriptsize{2}}
}
	child {
	node {\scriptsize{1}}}
    
;
\end{tikzpicture}
)
\begin{tikzpicture}[baseline=(2),level distance=5mm,sibling distance=4mm]
\node (1) {\scriptsize{2}}
    child {
	node {\scriptsize{2}}
}
	child {
	node {\scriptsize{2}}}
    
;
\end{tikzpicture}
}

Each object has to be drawn with the same probability $\frac{1}{14}$. This example points out that symmetries increase the difficulty to build an uniform sampler. That can be also observed by comparing Boltzmann sampler codes for labelled and unlabelled structures. In a naive approach, we could deal with in the classical Boltzmann model by duplicating $k$ times atoms, but we give here a more general and efficient point of view. In particular, we obtain a relatively compressed storage that we explain in a later section and we can postpone the draw of the colors to avoid the unnecessary coloration during the rejection phase.

A \textit{profiled  object} is by definition an object associated with a set-partition of its atoms. This concept is derived from Polya theory \cite{BLL}. We prove that this structure is lifted to all $k$-colorations. Roughly speeking, the atoms belonging the same partition-class could be interpreted as having the same colors. More precisely, for every $k$, a sampler for profiled objects can be specialised (by a good choice of the parameters) to become a sampler for any $k$-colored objects, so that, in a sense, our model unifies all $k$-colorations.
Espacially, profiled object samplers can be adapted to efficient sampling of colored object in approximative size. Profiles enable us to generate colored objects with Boltzmann methods even if the ordinary generating function is not analytic in 0 (which is an obstruction in classical Boltzmann theory). For instance, let $\mathcal{A}$ be the unlabelled constructible class Seq($\mathcal{Z}$), it is well known that the ordinary generating function is $1/(1-z)$. In this case, the ordinary generating function for the class $\tilde{\mathcal{A}}$ of colored objects is $\sum\limits_{k\geq 0}{n^nz^n}$ which is clearly not analytic in $0$. At first sight, it seems impossible to have a Boltzmann-type sampler for this class, since the parameter $x$ has to be in the disk of convergence of the generating function.


\smallskip
This paper is organised in four sections.
The first section is devoted to the definitions of various taggings for combinatorial classes. We define a general framework, in terms of language,  that allows to describe  unlabelled and labelled classes as well as less studied classes such as semilabelled or colored classes. We also recall the notion of constructible classes which is central for Boltzmann sampler theory.
The second section adresses the notion of profiled objects. We define the generating function associated to a profiled combinatorial class and express its relation with the generating functions of the associated $k$-colored classes. 

In the third section, we construct Boltzmann samplers for the profiled combinatorial constructible classes and we show how to use it to obtain samplers for the $k$-colored combinatorial constructible classes.

In the last section, we propose an approximate size sampler for the colored combinatorial constructible classes. This sampler is based on a filter that allows to transform a sampler for $n$-colored objects into a sampler for colored objects.
\section{Combinatorial Classes}
Combinatorial classes are very well studied and classical objects \cite{FS}. We propose in this section a general framework, in terms of language, that allows to describe extended tagged combinatorial classes. Our approach can be seen as an introduction to a very simplified species theory \cite{BLL}.  Let $\mathfrak{A}$ be the alphabet on the 7 following letters $\square \{ \} [ ] ( )$. The square is called the \textit{atom}. The other letters correspond to three different types of parentheses.

\newpage

We denote by $\mathcal{L}$ the language on $\mathfrak{A}$ defined as follows~:
\begin{itemize}

\item $\square, \{\}, (), []$ belong to $\mathcal{L}$
\item $\forall k\in \mathbb{N}^*, A_1,...,A_k\in \mathcal{L} \Rightarrow (A_1A_2...A_k)\in \mathcal{L}$, the \textit{sequences}
\item $\forall k\in \mathbb{N}^*, A_1,...,A_k\in \mathcal{L} \Rightarrow [A_1A_2...A_k]\in \mathcal{L}$, the \textit{cyclic sequences}
\item $\forall k\in \mathbb{N}^*, A_1,...,A_k\in \mathcal{L} \Rightarrow \{A_1A_2...A_k\}\in \mathcal{L}$, the \textit{multisets}
\end{itemize}

The elements of $\mathcal{L}$ are called the \textit{combinatorial proto-objects}.

The \emph{size} of a proto-object $A$ is the number of atoms contained in $A$. That is to say, the number of occurence of $\square$ in $A$. 
For instance, $(\{[\square\square](\square\square)\}\square)$ is a combinatorial proto-object of size 5.

Let $A$ be a proto-object, we denote $atom(A)$ the set of the occurences of $\square$ in $A$. In particular, $|atom(A)|$ is the size of $A$ and each element of $atom(A)$ corresponds exactly to one occurence of $\square$ in $A$.

%
%
At this stage, we just have a support but we need to explain the meaning of the cycles and the multisets by some equivalences between proto-objects. In order to do that, let $B$ be a set, a \textit{$B$-colored proto-object} is a pair $\langle A,f\rangle$ where $A$ is a word of $\mathcal{L}$ and $f$ is a function from $atom(A)$ to $B$. Let $a$ be an atom of $A$, $f(a)$ is called the \textit{color} of $a$. We denote by $\mathcal{L}_B $ the set of all $B$-colored proto-objects.

We define the equivalence $\equiv_B$ between $B$-colored proto-objects of $\mathcal{L}_B $ as follows~:

\begin{itemize}
\item Two $B$-colored squares are equivalent if and only if they are the same color~: $$\langle\square,f\rangle\equiv_B\langle\square,f'\rangle \Leftrightarrow f(\square)=f'(\square)$$
\item Two sequences are equivalent if their corresponding substructures are equivalent :
\begin{eqnarray*}
&\forall A_1,...,A_k,B_1,...,B_k\in \mathcal{L} , \langle (A_1A_2...A_k),f \rangle \equiv_B \langle (B_1B_2...B_k),f'\rangle\\
&\Leftrightarrow \\ 
&\forall i, f(Atom(A_i))=f'(Atom(B_i)) \mbox{ and } \langle A_i,f|_{Atom(A_i)}\rangle \equiv_{B} \langle B_i,f'|_{Atom(B_i)}\rangle 
\end{eqnarray*}

\item Two cycles are equivalent if their substructures are equivalent up to a circular shift :
\begin{eqnarray*}
&\forall A_0,...,A_{k-1},B_0,...,B_{k-1}\in \mathcal{L},\langle[A_0A_1...A_{k-1}],f\rangle\equiv_B \langle[B_0...B_{k-1}],f'\rangle\\
&\Leftrightarrow \\ 
&\exists n\forall i, f(Atom(A_i))=f'(Atom(B_{i+n \bmod k}))\\
& \mbox{ and }\\ 
&\langle A_i,f|_{Atom(A_i)}\rangle\equiv_{B} \langle B_{i+n \bmod k},f'|_{Atom(B_{i+n \bmod k})}\rangle
\end{eqnarray*}

\item Two sets are equivalent if their substructures are equivalent up to permutations (we denote by $S_k$ the symmetric group of order $k$) :
\begin{eqnarray*}
&\forall A_0,...,A_{k-1},B_0,...,B_{k-1}\in \mathcal{L},\langle\{A_0...A_{k-1}\},f\rangle\equiv_B \langle\{B_0...B_{k-1}\},f'\rangle\\
&\Leftrightarrow \\ 
&\exists \sigma\in S_k,\ \forall i, f(Atom(A_i))=f'(Atom(B_{\sigma(i)}))\\ 
&\mbox{ and }\\ 
&\langle A_i,f|_{Atom(A_i)}\rangle\equiv_{B} \langle B_{\sigma(i)},f'|_{Atom(B_{\sigma(i)})}\rangle
\end{eqnarray*}
\end{itemize}

The set $\mathcal{O}_B$ of \emph{$B$-colored combinatorial objects} is $\mathcal{L}_B /\equiv_B$. Now, we can define the fundamental notion of combinatorial classes as follows : 

\begin{definition}
A \textit{$B$-colored combinatorial class} is a sub-multiset of $\mathcal{L}_B/\equiv_B$ with a finite number of object of each size.  
\end{definition}

\newpage

We opt in this definition for a multiset version which avoid the traditionnal problem of the disjoint union. The union $A\cup B$ of two multisets is the multiset obtained with all the existing elements and where the multiplicity of $a$ in $A\cup B$ is the sum of the multplicity $a$ in $A$ and the multiplicity of $a$ in $B$.

The \emph{size} of a $B$-colored combinatorial object $A$ is the size of any proto-object which represents it. 

Every object of size 0 is called \textit{neutral object}. For instance, $\{()()[]\}$ is a neutral object. Clearly, $\mathcal{O}_B$ contains an infinite number of neutral objects. we carefully avoid the confusion between the empty combinatorial class $\emptyset$ and a combinatorial class reduced to a neutral element $\epsilon$.

For example, the following $\{1,2\}$-colored combinatorial class $\mathcal{C}$ contains 3 objects. The both objects of size 3 are identical. $$\mathcal{C}=\displaystyle \pg \left( \B{1}\right),\left \{ \left [\B{2}\B{1}  \right ]\B{1} \right\},\left \{\B{1}  \left [\B{1}\B{2}  \right ]\right\} \pd$$


%

Now, we can add some restrictions on $f$ to obtain different types of \emph{tagged} combinatorial objects. Some of them are classical (unlabelled, labelled), the other ones are more unusual but reasonable.

\begin{figure}[htbp]\small
\begin{center}
\begin{tabular}{|l|l|l|}
\hline
Function & Codomain & Labelling
\\
\hline
\hline
one-to-one map &  $\{1,...,n\}$ for an object of size $n$ & labelled 
\\
\hline
constant map & for instance $\{1\}$ & unlabelled 
\\
\hline
quasi-constant map & $\{0,1\}$ ($\exists!y\forall x\neq y f(x)=1$ and $f(y)=0$) & pointed 
\\
\hline
surjective map &$\{1,...,k\}$ with $k\leq n$ &  semi-labelled 
\\
\hline
arbitrary map &$\{1,...,n\}$ for an object of size $n$ &  colored 
\\
\hline
arbitrary map &$\{1,...,k\}$ with $k$ fixed &  $k$-colored 
\\
\hline
\end{tabular}
\end{center}
\caption{\label{tab:const} Different labellings.}
\end{figure}

For each of these tagged objects, we can easily define the notion of combinatorial classes.

We essentially study in this paper the notion of colored and $k$-colored classes. The semilabelled classes are the topic of another work.

Let $\mathcal{A}$ be a combinatorial (unlabelled) class, we denote by $\mathcal{A}_n$ the set of objects of size $n$.
\subsection{Constructible combinatorial classes}

A \textit{constructible combinatorial class} $\mathcal{A}(\mathcal{Z})$ is a class which can be build from  neutral classes and an atomic class $\mathcal{Z}$ by using some of these builders ($+$,$\times$,$Seq$,$MSet$,$PSet$,$Cyc$,...). (the builders $Seq$, $MSet$, $PSet$, $Cyc$ cannot be applied to classes having neutral objects.)
For instance, $MSet(\mathcal{Z}\times Seq(\mathcal{Z}))$. We also accept the recursive construction : $\mathcal{T}=f(\mathcal{T})$.\\
These constructions are exactly the ones that we recurrently find in modern theories of combinatorial analysis. We do not develop here the precise grammar of building. This can be found for instance in \cite{FS}.

A \textit{colored constructible combinatorial class} (resp. \textit{$k$-colored constructible combinatorial class}) is build from a  constructible combinatorial class $\mathcal{A}(\mathcal{Z})$ as follows~: $\mathcal{\tilde{A}}=\bigcup\limits^{\infty}_{n=0}\mathcal{A}_n(\mathcal{Z}_1+\ldots+\mathcal{Z}_n)$ (resp. $\mathcal{A}^{[k]}=\mathcal{A}(\mathcal{Z}_1+\ldots+\mathcal{Z}_k)$). The unique atom of the atomic class $\mathcal{Z}_k$ corresponds to an atom colored with the color $k$. Strictly speaking, $\mathcal{\tilde{A}}$ and $\mathcal{A}^{[k]}$  are not combinatorial tagged classes as defined in the previous section. But they can be obviously identify by expressing the colors in terms of function on the atoms.

\newpage

In the sequel, we denote identically the colored constructible combinatorial class and its associated unlabelled combinatorial class. For instance, we could speak about the colored constructible combinatorial class $MSet(\mathcal{Z}\times Seq(\mathcal{Z}))$ to say that we consider the colored constructible combinatorial class build from the unlabelled constructible class $MSet(\mathcal{Z}\times Seq(\mathcal{Z}))$.





\section{Profiled combinatorial classes}

Let $A$ be a set, a subset $S$ of $\mathcal{P}(A)$ (the set of all parts of $A$) is called a \textit{set-partition} of $A$ if and only if the elements of $S$ are disjoint, $S$ does not contain $\emptyset$ unless that $A=\emptyset$ and the union of elements of $S$ are equal to $A$. 

A \emph{profiled combinatorial object} $\langle a, \sigma_a\rangle$ is a combinatorial object $a$ equiped with a set-partition $\sigma_a$ of $atom(a)$ which represents its \emph{profile} or \emph{symmetry}. 

Remark : In the classical Polya theory \cite{BLL}, we deal with an automorphism $s$ of $atom(a)$ but this is overdimensioned according to what we need. We only want to mark the undistinguished atoms, a partition is clearly sufficient.

\begin{definition}
A \textit{profiled combinatorial class} is a multiset of profiled combinatorial objects with a finite number of object of each size.  
\end{definition}

Now, we can define the product and the diagonal of profiled combinatorial objects as follows :

\textbf{Product (concatenation).}
We define the symmetry $\sigma_m$ of $m=b\times c$ as $\sigma_b\cup \sigma_c$. In other words, $\langle b\times c, \sigma_b\cup \sigma_c\rangle$.

\textbf{Diagonal.}
We define $\delta_ka=a...a$ where $a$ is repeated $k$ times.

The symmetry of $\delta_ka$ is defined as follows : Let $\{P_1,...,P_m\}$ be the symmetry of $a$ then the symmetry of $\delta_ka$ is $\{\bigcup\limits_{i}{P_1^i},...,\bigcup\limits_{i}{P_m^i}\}$   where $P_1^i$ corresponds to $P_i$ but on the atoms of the $i$-th occurence of $a$. 
For instance, let $\langle a, \sigma_a\rangle=\langle(\{[\square\square](\square\square)\}\square),\{\{1,3\},\{2,5\},\{4\}\}\rangle$ (the first $\square$ is linked to the third $\square$...), then $\delta_2\langle a, \sigma_a\rangle$ is $\langle aa,\{\{1,3\}\cup\{6,8\},\{2,5\}\cup\{7,10\},\{4\}\cup\{9\}\}\rangle$ 




A coloration of a profiled object of size $n$ and symmetry $\sigma$  is \textit{consistant} if on each subset of $\sigma$ every atom has the same color. We also say that a profile is \textit{$c$-admissible} for a coloration $c$ if each part of $\sigma$ is monocolor.

\subsection{Constructible profiled combinatorial classes}
the constructible profiled combinatorial classes are exactly the classes that we can buld with the following list of operators :
\begin{itemize}
\item \textbf{Atomic class} : $\mathcal{Z}=\{\square\}$, the set-partition is reduce to the singleton.
\item \textbf{Union} : $\mathcal{A}\cup\mathcal{B}$ is the multiset-union of the classes. 
\item\textbf{Product} : $\mathcal{A}\times\mathcal{B}=\{\langle a\times b, \sigma_a\cup \sigma_b\rangle ; a\in \mathcal{A}, b\in\mathcal{B}\} $
\item \textbf{Sequence} : $Seq(\mathcal{A})=\bigcup\limits_{i=0}^\infty\mathcal{A}^i$
\item \textbf{Diagonal} : $\Delta_k(\mathcal{A})=\{\delta_ka; a\in \mathcal{A}\}$
\item \textbf{MultiSet} :\\ $MSet(\mathcal{A})=\left \{ \{B\}; B \mbox{ a concatenation of diagonals of elements in } \mathcal{A}  \right\}$
\item \textbf{Cycle} :\\ $Cyc(\mathcal{A})=\left \{ [B\,]; B \mbox{ a diagonal of concatened elements of } \mathcal{A}  \right\}$
\end{itemize}

A classical tool associated to combinatorial classes is the notion of generating functions. We explain below the building rules for the generating multivariate functions on the infinite many variables $s_1,s_2,...$ associated to constructible profiled combinatorial classes. The variable $s_i$ are going to express a set of $i$ linked atoms.

\newpage

For example, this profiled general non-planar tree of size $14$ :

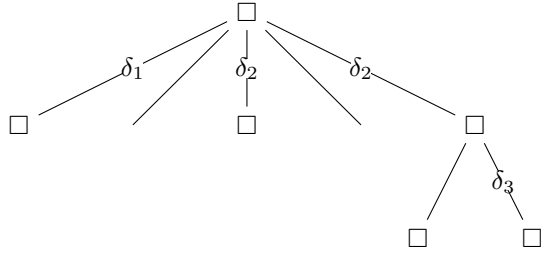
\begin{center}
\begin{tikzpicture}

\node (1) {$\square$} 
  child {  
    node (2) {$\square$} 
    edge from parent 
    node[fill=white,-|,inner sep=-0.8pt] {$\delta_1$}
}
child 
child {
   node (4_1) {$\square$}
    edge from parent
      node[fill=white,-|,inner sep=-0.8pt] {$\delta_2$}
  }
child 
child {
    node (5_1) {$\square$}    
   child {
    node (6_1) {$\square$}
    }
   child{
     node (8_11) {$\square$}
     edge from parent 
     node[fill=white,-|,inner sep=-0.9pt] {$\delta_3$}
   }
   edge from parent 
   node[fill=white,-|,inner sep=-0.8pt] {$\delta_2$}
}
;

\end{tikzpicture}
\end{center}

has $s_1^2.s_2^3.s_6$ as associated monomial.

\bigskip

 In fact, 
the multivariate generating function associated to constructible profiled combinatorial classes can be defined inductively as follows :
\begin{itemize}
\item \textbf{Neutral class} : $F_\mathcal{E}=1$
\item \textbf{Atomic class} : $F_\mathcal{Z}=s_1$
\item \textbf{Sum} : $F_{\mathcal{A}+\mathcal{B}}=F_\mathcal{A}+F_\mathcal{B}$
\item\textbf{Product} : $F_{\mathcal{A}\times\mathcal{B}}=F_\mathcal{A}\times F_\mathcal{B}$
\item \textbf{Sequence} : $F_{Seq(\mathcal{A})}=\dfrac{1}{1-F_\mathcal{A}}$
\item \textbf{Diagonal} : $F_{\Delta_k(\mathcal{A})}(s_1,s_2,...)=F_{\mathcal{A}}(s_k,s_{2k},...)$
\item \textbf{MultiSet} : $F_{MSet(\mathcal{A})}=\exp{\sum_{k\geq 1} \frac{1}{k}F_{\Delta_k(\mathcal{A})}}$
\item \textbf{Cycle} : $F_{Cyc(\mathcal{A})}=\sum_{k\geq 1} -\frac{\varphi(k)}{k} \ln(1-F_{\Delta_k(\mathcal{A})})$
\end{itemize}

Let us notice that $F_{\mathcal{A}}$ is not an enumerative generating function for the profiled objects in $\mathcal{A}$. But it has the following very interesting property. Let us put $f_{\mathcal{A}}(x,t)=F_{\mathcal{A}}(t.x,t.x^2,...,t.x^k,...).$ It is a classical result of Polya theory that $[x^n]f_{\mathcal{A}}(x,t)=|\mathcal{A}_n(\mathcal{Z}_1+\ldots+\mathcal{Z}_t)|.$ The subtitution of $s_i$ by $t.x^i$ only says that the set of $i$ linked atoms represented by $s_i$ can exactly be colored in $t$ different ways and it must be considered as an object of size $x^i$. So, the function $f_{\mathcal{A}}(x,t)$ is the ordinary generating function for the $t$-colored combinatorial class $\mathcal{A}^{[t]}$. In a sense, the generating function of $F_{\mathcal{A}}$ contains the informations of all the generating functions for the $t$-colored combinatorial classes. In fact, this multivariate generating function is generally called cycle index sum and it is also fundamental in \cite{BFKV}. For instance, consider $\mathcal{S}=Mset(\mathcal{Z})$, the first terms of the profiled generating function are $F_{MSet(\mathcal{Z})}=1+s_1+s_1^2/2+s_2/2+s_1^3/6+s_1s_2/2+s_3/3+...$. To obtain the enumerative generating function for the 3-colored $Mset(\mathcal{Z})$, it suffices to replace $s_i$ by $3x^i$. So, we find $F_{MSet(\mathcal{Z})}=1+3x+6x^2+10x^3+...$

\smallskip
Remark : In constructive profiled classes, every symetries are only obtained by the diagonal operator. So, in the implementation, we do not need to stock the partition, we simply keep the diagonals unexpanded. That's why our profiled samplers return only an ``object'' $a$, instead of a couple $\langle a, \sigma \rangle$ as in our proofs. Moreover, this choice induces a storage gain, as big as the object has symmetries.



\section{Probability and Boltzmann Samplers}
The principle of Boltzmann samplers can be describe as follows : Let $\mathcal A$ be a combinatorial unlabelled class and $A(x)$ its ordinary generating function. Consider a non negative real number $x$ in the convergence disk of $A(x)$, a Boltzmann sampler $\Gamma\mathcal A$ is a random generator that draws each object $a\in\mathcal A$ with probability $\mathbb{P}(a)=\frac{x^{|a|}}{A(x)}.$ There are simple rules, described in \cite{DFLS} to build automatically Boltzmann samplers from the combinatorial specification of $\mathcal A$. 

For our purpose, we want to draw $t$-colored objects coming from a  specified class $\mathcal{A}(\mathcal{Z})$ with the following Boltzmann probability $\mathbb{P}_{t,x}(a)=\dfrac{x^{|a|}}{f_{\mathcal{A}}(x,t)}$. 

We produce samplers for (non colored) profiled objects with the following Boltzmann distribution : the probability to obtain a profiled object  $\langle a,\sigma\rangle$ such that $\sigma$ has $n_i$ parts of size $i$ is  $\dfrac{[s_1^{n_1}s_2^{n_2}...]F_{\mathcal{A}}(s_{(n)\in\mathbb{N}^*}).s_1^{n_1}s_2^{n_2}...}{F_{\mathcal{A}}(s_{(n)\in\mathbb{N}^*})}$ 
. We then prove that after a consistant coloration and a good choice of the parameters, we also obtain a Boltzmann sampler for the $t$-colored objects. The notion of powersets is not very relevant on profiled combinatorial classes. Indeed, in a profiled object, some atoms are undistinguished, but there is no information on distinguished atoms, they can be the similar or not. Below, a tabular with the classical distributions of probability and the design rules for basic sampler constructions where $\Gamma\mathcal A$ designs a Boltzmann sampler for the class $\mathcal A$.

\begin{figure}[htbp]\small
\begin{center}
\begin{tabular*}{\linewidth}{lll}
\hline
Distribution & Notation & Definition\\
\hline
\hline
Bernoulli &   Bern$(p)$ & $\mathbb{P}(0) = 1-p$ and $\mathbb{P}(1) = p$ (with $0 \leq p \leq 1$)
\\
Geometric & Geom($\lambda$) & $\mathbb{P}(k) = \lambda^k(1 - \lambda)$ (with $k\in \mathbb{N}$ and $0 \leq \lambda < 1$)
\\
Poisson & Pois($\lambda$) &  $\mathbb{P}(k) =e^{-\lambda}\dfrac{\lambda^k}{k!}$ (with $k\in \mathbb{N}$ and $\lambda\in \mathbb{R^{+*}}$)
\\
Positive Poisson & Pois$_{\geq0}$($\lambda$) &  $\mathbb{P}(k) =\dfrac{1}{e^{\lambda}-1}\dfrac{\lambda^k}{k!}$ (with $k\in \mathbb{N^*}$ and $\lambda\in \mathbb{R^{+*}}$)
\\

\hline
\end{tabular*}
\end{center}
\caption{\label{tab:const} Distributions of use in Boltzmann sampling.}
\end{figure}

In the sequel, the symbol $\leftarrow$ designs the affection.
\begin{itemize}
\item \textbf{Neutral class} : $\Gamma \mathcal{E}_{(s_1,s_2,...)}:=$ return $()$
\item \textbf{Atomic class} : $\Gamma \mathcal{Z}_{(s_1,s_2,...)}:=$ return $\square$
\item \textbf{Sum} : $\Gamma (\mathcal{A+B})_{(s_1,s_2,...)}:=$\\
Draw $X$ following a Bernoulli law of parameter $Bern(\dfrac{F_{\mathcal{A}}{(s_1,s_2,...)}}{F_{\mathcal{A}+\mathcal{B}}{(s_1,s_2,...)}})$\\
If $X=1$ return $\Gamma \mathcal{A}_{(s_1,s_2,...)}$ else return $\Gamma \mathcal{B}_{(s_1,s_2,...)}$
\item \textbf{Product} : $\Gamma (\mathcal{A \times B})_{(s_1,s_2,...)}:=$\\
Let $A\leftarrow \Gamma \mathcal{A}_{(s_1,s_2,...)}$ and $B \leftarrow \Gamma \mathcal{B}_{(s_1,s_2,...)}$ \\
return $(AB)$
\item \textbf{Diagonal} : $\Gamma {\Delta_k(\mathcal{A})}_{(s_1,s_2,...)}:=$\\
let $A \leftarrow\Gamma{\mathcal{A}}_{(s_k,s_{2k},...)}$\\ 
return $ \delta_kA$.
\item \textbf{Sequence} : See Algorithm 1.
\end{itemize}






\begin{itemize}
\item \textbf{MultiSet} : See Algorithm 2.
\end{itemize}
\begin{itemize}
\item \textbf{Cycle} : See Algorithm 3. 
\end{itemize}
\begin{proposition}
 The previous samplers are valid Botzmann samplers for profiled objects.
\end{proposition}

\begin{algorithm}[H]
\caption{$\Gamma Seq(\mathcal{A})_{(s_1,s_2,...)}$}

\KwIn{the parameters $s_1,s_2...$}

\KwOut{a sequence.}

Draw $k$ following the geometric law $Geom(F_{\mathcal{A}}{(s_1,s_2,...)}).$\\

\For{$i$ \emph{\textbf{from }} $1$ \textbf{\emph{to }} $k$}{
$A_i\leftarrow \Gamma \mathcal{A}_{(s_1,s_2,...)}$\\
}
\Return $(A_1...A_k)$.
\end{algorithm}

\smallskip

\begin{algorithm}[H]
\caption{$\Gamma MSet(\mathcal{A})_{(s_1,s_2,...)}$}

\KwIn{the parameters $s_1,s_2...$}

\KwOut{a multiset.}

let $K$ be a random variable in $\mathbb{N}^*$ verifying $\mathbb{P}(K\leq k)=\prod\limits_{j>k}\exp(\frac{1}{j}F_{\mathcal{A}}(s_j,s_{2j},...))$\\
Draw $k$ following the law of $K$.\\
$S\leftarrow \epsilon$\\
\eIf{$k=1$}{
Draw $q$ following a Poisson law  of parameter $F_{\mathcal{A}}{(s_1,s_{2},...)}$.\\
\For{$i$ \emph{\textbf{from }} $1$ \textbf{\emph{to }} $q$}{
$A_i\leftarrow \Gamma\mathcal{A}_{(s_1,s_2,...)}$\\
$S\leftarrow Concat(S, A_i)$.\\
}
\Return $\{S\}$\\}
{\For{$j$ \emph{\textbf{from }} $1$ \textbf{\emph{to }} $k-1$}{
Draw $q$ following a Poisson law  of parameter $\frac{1}{j}F_{\mathcal{A}}{(s_j,s_{2j},...)}$.\\
\For{$i$ \emph{\textbf{from }} $1$ \textbf{\emph{to }} $q$}{
$A_i\leftarrow \Gamma\Delta_j(\mathcal{A})_{(s_1,s_2,...)}$\\
$S\leftarrow Concat(S, A_i)$.\\
}}
Draw $q$ following a Poisson$_{\geq 1}$ law  of parameter $\frac{1}{k}F_{\mathcal{A}}{(s_k,s_{2k},...)}$.\\
\For{$i$ \emph{\textbf{from }} $1$ \textbf{\emph{to }} $q$}{
$A_i\leftarrow \Gamma\Delta_k(\mathcal{A})_{(s_1,s_2,...)}$\\
$S\leftarrow Concat(S, A_i)$.\\
}
\Return $\{S\}$\\}
\end{algorithm}









\smallskip

\begin{algorithm}[H]
\caption{$\Gamma Cyc(\mathcal{A})_{(s_1,s_2,...)}$}

\KwIn{the parameters $s_1,s_2...$}

\KwOut{a cycle.}

let $K$ be a random variable in $\mathbb{N}^*$ verifying $\mathbb{P}(K= k)=- \frac{1}{F_{Cyc(\mathcal{A})}} \frac {\varphi(k)}{k} \  \ln(1-F_{\Delta_k(\mathcal{A})}) $\\
Draw $k$ following the law of $K$.\\

let $L$ be a random variable in $\mathbb{N}^*$ verifying $\mathbb{P}(L= l)=-\frac{( F_{\Delta_k(\mathcal{A})})^{l}}{l} \frac{1}{\ln(1-F_{\Delta_k(\mathcal{A})})} $\\
Draw $l$ following the law of $L$.\\

$M\leftarrow\epsilon$\\

\For{$i$ \emph{\textbf{from }} $1$ \textbf{\emph{to }} $l$}{
$A_i\leftarrow \Gamma\mathcal{A}_{(s_k,s_{2k},...)}$\\
$M\leftarrow Concat(M, A_i)$.\\
}
\Return $[\delta_kM]$.
\end{algorithm}

\begin{proof}[Proof of the proposition 3.1]
We only prove that for the sampler for MSet and Cyc. The other ones are easy exercices. Let us begin with the MSet sampler.
 First, suppose that during the execution of the algorithm, we have drawn $k=1$. So, $\mathbb{P}(K=1)=\exp(F_{\mathcal{A}}(s_1,s_2,...))$. After that, we have choosen $q$ with probability : $$\mathbb{P}(q)=\exp(-F_{\mathcal{A}}(s_1,s_2,...)).\dfrac{F_{\mathcal{A}}(s_1,s_2,...)^q}{q!}$$ and we have drawn $q$ objects in $\mathcal{A}$, each of them has profile $\sigma$ such that the number of parts of $\sigma \mbox{ of cardinal } i \mbox{ is } n_i$ with probability :
 $$\mathbb{P}_\mathbf{n}=\dfrac{[\mathbf{s^n}]F_{\mathcal{A}}(s_{(n)\in\mathbb{N}^*}).\mathbf{s^n}}{F_{\mathcal{A}}(s_{(n)\in\mathbb{N}^*})}.$$ where we denote $s_1^{n_1}s_2^{n_2}...$ by $\mathbf{s^n}$. So, the probability to obtain an object of profile $\sigma$ such that the number of parts of $\sigma \mbox{ of cardinal } i \mbox{ is } n_i$ (when we draw $k=1$) is the product of all these terms.
 This can be simplify in : $$\dfrac{
\mathbf{s^n}}{\prod\limits_{j=1}^{\infty} \dfrac{1}{j}F_{\mathcal{A}}(s_j,s_{2j},...)}\sum\limits_{q=0}^{\infty}{\sum\limits_{(\mathbf{m_1},...,\mathbf{m_q});\sum\limits_{d}{\mathbf{m_d}}=\mathbf{n}}{\dfrac{\prod\limits_{d=1}^{q}{[\mathbf{s^{m_d}}]F_{\mathcal{A}}(s_{(n)\in\mathbb{N}^*})}}{q!}}}.$$

Now, for $k>1$,  in the same way, we can show that the probability is $$\sum\limits_{k>1}\sum\limits_{q_1> 0,...,q_k>0}{
\mathop{\sum\limits_{(\mathbf{m}_{1,1},...,\mathbf{m}_{k,q_k});}}_ {\sum\limits_{i=1}^{k}\sum\limits_{j=1}^{q_i}{\mathbf{m}_{i,j}}=\mathbf{n}}
 {\dfrac{1}{1^{q_1}q_1!...k^{q_k}q_k!}.\dfrac{\prod\limits_{p=1}^{k}{\prod\limits_{d=1}^{q_p}{ [\mathbf{s^{m_{p,d}}}]F_{\mathcal{A}}(s_{(n)\in\mathbb{N}^*}).\mathbf{s}^{\mathbf{m}_{p,d}}}}}{\prod\limits_{j=1}^{\infty}{ \dfrac{1}{j}F_{\mathcal{A}}(s_j,s_{2j},...) 
}}}}.$$

So, using this sampler, the probability to obtain an object of profile $\sigma$ such that the number of parts of $\sigma \mbox{ of cardinal } i \mbox{ is } n_i$ is
$$\sum\limits_{k\geq 1}\sum\limits_{q_1> 0,...,q_k>0}{
\mathop{\sum\limits_{(\mathbf{m}_{1,1},...,\mathbf{m}_{k,q_k});}}_{\sum\limits_{i=1}^{k}\sum\limits_{j=1}^{q_i}{\mathbf{m}_{i,j}}=\mathbf{n}} {\dfrac{1}{1^{q_1}q_1!...k^{q_k}q_k!}.\dfrac{
\mathbf{s}^{\mathbf{n}}\prod\limits_{p=1}^{k}{\prod\limits_{d=1}^{q_p}{[\mathbf{s^{m_{p,d}}}]F_{\mathcal{A}}(s_{(n)\in\mathbb{N}^*})}}} {\prod\limits_{j=1}^{\infty} \dfrac{1}{j}F_{\mathcal{A}}(s_j,s_{2j},...)}}}.$$
 Indeed, this is all the ways to obtain an object of such a profile. 

Now, 
$[\mathbf{s^{n}}]F_{MSet(\mathcal{A})}(s_1,s_2,...)=\sum\limits_{\sum\limits_{j}{n_j}=n}{\prod\limits_{j}{[\mathbf{s^{n_j}}]\exp(\frac{1}{j}F_{\mathcal{A}}(s_j,s_{2j},...))}}$. Moreover, $$[\mathbf{s^{n_j}}]\exp(\frac{1}{j}F_{\mathcal{A}}(s_j,s_{2j},...))=\sum\limits_{\alpha=0}^{\infty}{\frac{1}{\alpha!}\sum\limits_{\sum\limits_{p=1}^{\alpha}n_{j,p}=n_j}{\prod\limits_{p=1}^{\alpha}{[\mathbf{s}^{\mathbf{n}_{j,p}}]\dfrac{F_{\mathcal{A}}(s_j,s_{2j}...)}{j}}}}.$$ So, $$[\mathbf{s^{n}}]F_{MSet(\mathcal{A})}(s_1,s_2,...)=\sum\limits_{\sum\limits_{j=1}^{\infty}{n_j}=n}{\prod\limits_{j=1}^{\infty}{\sum\limits_{\alpha=0}^{\infty}{\frac{1}{\alpha!}\sum\limits_{\sum\limits_{p=1}^{\alpha}n_{j,p}=n_j}{\prod\limits_{p=1}^{\alpha}{[\mathbf{s}^{\mathbf{n}_{j,p}}]\dfrac{F_{\mathcal{A}}(s_j,s_{2j}...)}{j}}}}}}.$$ We have to change the order of the sums and products. 
We use two times that $\prod\limits_{j=1}^{\infty}\sum\limits_{\alpha\in A}{a_{\alpha,j}}=\sum\limits_{\mathbf{\alpha}\in A^\infty}{\prod\limits_{j=1}^{\infty } }{a_{\alpha_j,j}} $
 , and we obtain that the probability to draw an object of profile $\sigma$
is $\dfrac{\mathbf{s}^\mathbf{n}.[\mathbf{s}^\mathbf{n}]F_{MSet(\mathcal{A})}(s_1,s_2,...)}{F_{MSet(\mathcal{A})}(s_1,s_2,...)}.$ 

Now, we analyse the sampler for Cyc. 
Suppose that the algorithm returns an object of size $n$ and profile $\sigma$ such that the number of parts of $\sigma \mbox{ of cardinal } i \mbox{ is } n_i$. 
The probability $p_{\mathbf{n}}$ to draw such an object is 
$$\sum\limits_{k.l=n}\sum\limits_{(\mathbf{m}_1,...,\mathbf{m}_l);\sum\limits_{d=1}^l{\mathbf{m}_d}=\mathbf{n}}\mathbb{P}(K=k).\mathbb{P}(L=l).\prod_{i=1}^lp_{\mathbf{m}_l}$$
with $\mathbb{P}(K= k)=- \dfrac{1}{F_{Cyc(\mathcal{A})}(s_{(n)\in\mathbb{N}^*})}
\dfrac {\varphi(k)}{k} \  \ln(1-F_{\Delta_k(\mathcal{A})}) $, $$\mathbb{P}(L=
l)=-\dfrac{( F_{\Delta_k(\mathcal{A})}(s_{(n)\in\mathbb{N}^*}))^{l}}{l}
\dfrac{1}{\ln(1-F_{\Delta_k(\mathcal{A})}(s_{(n)\in\mathbb{N}^*}))} $$ and $\forall
i\in \{0...l\}$, $p_{\mathbf{m}_i} =
\dfrac{[\mathbf{s}^{\mathbf{m}_i}]F_{\Delta_k(\mathcal{A})}(s_{(n)\in\mathbb{N}^*}).\mathbf{s}^{\mathbf{m}_i}}{F_{\Delta_k(\mathcal{A})}(s_{(n)\in\mathbb{N}^*})}$\\
So, $$p_{\mathbf{n}}=\sum\limits_{k.l=n}\sum\limits_{(\mathbf{m}_1,...,\mathbf{m}_l);\sum\limits_{d=1}^l{\mathbf{m}_d}=\mathbf{n}}
\dfrac{1}{F_{Cyc(\mathcal{A})}(s_{(n)\in\mathbb{N}^*})} \dfrac{\varphi(k)}{k}
\dfrac{1}{l}
\prod_{i=1}^l[\mathbf{s}^{\mathbf{m}_i}]F_{\Delta_k(\mathcal{A})}(s_{(n)\in\mathbb{N}^*}).\mathbf{s}^{\mathbf{m}_i}$$ $$=
\dfrac{[\mathbf{s}^{\mathbf{n}}]F_{Cyc(\mathcal{A})}(s_{(n)\in\mathbb{N}^*}).\mathbf{s}^{\mathbf{n}}}{F_{Cyc(\mathcal{A})}(s_{(n)\in\mathbb{N}^*})}$$ 
\end{proof}

\section{Boltzmann Samplers for $t$-colored classes and colored classes}
\subsection{The $t$-colored Boltzmann Samplers}
In this section, we mention how to obtain from the sampler for profiled objects the samplers for the $t$-colored objects. 
\begin{proposition}
 Let us substitute $s_i$ by $t.x^i$ and generate a profiled object $\langle a,\sigma\rangle$ as previously. For each part $p$ of $\sigma$, draw uniformaly a color $c_p$ in $\{1,...,t\}$ and let us assign the atoms in $p$ with the color $c_p$.

This sampler is a valid Boltzmann sampler for the $t$-colored objects in the class $\mathcal{A}$.
\end{proposition}
\begin{proof}
 Let us consider a $t$-colored object $a$ drawn by this sampler. Suppose that $a$ is of size $n$ and have $n_i$ atoms colored with the color $i$. The probability to obtain $a$ is the sum on the $a$-admissible profiled objects times the probability to put the good colors. So, we have :
$\mathbb{P}(a)= \sum\limits_{\sigma \mbox{\tiny{admissible profile}}}{\dfrac{\prod{(t.x^i)^{n_i}}}{F_{\mathcal{A}}(t.x,t.x^2,...)}.\dfrac{\kappa_\sigma.[\mathbf{s}^\mathbf{n}]F_{\mathcal{A}}(t.x,t.x^2,...)} {t^{\sum{n_i}}}}=\dfrac{x^{|a|}}{f_{\mathcal{A}}(x,t)}.\sum\limits_{\sigma \mbox{\tiny{admissible profile}}}{\kappa_\sigma.[\mathbf{s}^\mathbf{n}]F_{\mathcal{A}}(t.x,t.x^2,...)}$ where $\kappa_\sigma$ is the number of $\equiv_B$-equivalent coloration that corresponds to the color of $a$.\\ It can be proved that $\sum\limits_{\sigma \mbox{\tiny{admissible profile}}}{\kappa_\sigma.[\mathbf{s}^\mathbf{n}]F_{\mathcal{A}}(t.x,t.x^2,...)}=1$.

Another proof of this proposition is just to observe that after the substitution of $s_i$ by $t.x^i$, our sampler is syntaxically identical to a classical Boltzmann sampler for $k$-colored objects, except that we have postponed the Bernoulli choice of the color by the informations stored in the profile. 
\end{proof}

So, we have valid Boltzmann samplers for the $t$-colored objects. We can use it for exact size generation as well as approximate size
sampling. 

\subsection{Approximate-size Boltzmann Samplers for size-colored objects}
Now, we want to use profiled object samplers to generate colored objects. Suppose that we want to generate a colored object of size $n$ in a combinatorial class $\mathcal{A}$. First, we have to solve the following equation to find the parameter $x_0$ which targets correctly the samplers :
 $x.\dfrac{\frac{\partial}{\partial x} f_{\mathcal{A}}(x,n)}{f_{\mathcal{A}}(x,n)}=n$. We draw a profiled object with the parameters $s_i=n.x_0^i$. Now, if we color this object with $n$ colors, the sampler is uniform only  on the class $\mathcal{A}^{[n]}_{n}$ that is to say when we have drawn a profiled object of size $n$. We unbiais the approximate size generator by adding some rejections. The following lemma explains how we filter the sampler in order to correct it. 
\begin{lemma}\label{lemma}
 Let $A=\{x_1,...,x_k\}$ be a set of objects and $\Gamma$ a sampler for $A$ that draws the object $x_i$ with probability $p_i\neq 0$. Let $(p'_1,...,p'_k)$ another distibution of probability on $A$, and suppose that the $\dfrac{p'_i}{p_i}$ forms a decreasing sequence. Consider the following sampler $\Gamma'$ :\\

\begin{algorithm}[H]
\caption{$\Gamma'$}

\KwIn{the distribution of $p'_i$}

\KwOut{an object.}
Draw an object $x$ with $\Gamma$\\
Let $i$ such that $x=x_i$\\
Draw X following a Bernoulli law of parameter $\dfrac{p'_i.p_1}{p_i.p'_1}$\\
\eIf{$X=1$}{\Return $x$}{restart $\Gamma'$}

\end{algorithm}

Then $\Gamma'$ draws the object $x_i$ with probability $p'_i$. Moreover, $\Gamma'$ calls on average $\dfrac{p'_i}{p_i}$ times $\Gamma$ to obtain an object.
\end{lemma}

\begin{proof}
First, observe that the parameter $\dfrac{p'_i.p_1}{p_i.p'_1}$ is in $[0,1]$.
 Put $x_i=\dfrac{p'_i.p_1}{p_i.p'_1}$, the probability to obtain $x_i$ is $\sum\limits_{n=0}^{\infty}{(\sum\limits_{j}{p_j\bar{x}_j})^n}p_ix_i=\dfrac{p_ix_i}{1-\sum\limits_{j}{p_j\bar{x}_j}}=\dfrac{p_ix_i}{\sum\limits_{j}{p_jx_j}}$. In this expression, substitute $x_i$ by $\dfrac{p'_i.p_1}{p_i.p'_1}$, we get $\dfrac{p_i\dfrac{p'_i.p_1}{p_i.p'_1}}{\sum\limits_{j}{p_j\dfrac{p'_j.p_1}{p_j.p'_1}}}=\dfrac{p'_i}{\sum\limits_{j}{p'_j}}=p'_i.$
The expected time is $\sum\limits_{n=0}^{\infty}{(n+1)(\sum\limits_{j}{p_j\bar{x}_j})^n}p_ix_i$. A similar calculation shows that this is equal to $\dfrac{p'_i}{p_i}$.
\end{proof}

 Now, consider the problem for sampling for colored objects of size around $n$, a sampler is described in Algorithm 5. In a first time, we propose to observe its running on an example  :\\

\begin{example} A generation for a  colored general tree of specification $T=\mathcal{Z}+\mathcal{Z}MSet(T)$ by Algorithm 5, with parameters $14$, $x_0$.

First, a profiled tree $A$ is drawn by Algorithm 2 with $s_i = 14.x_0^i$ :\\
We draw k=2.\\
For j=1, we draw q=1, consequently we draw an only profiled object using $\Gamma\Delta_{1}T_{(s_1,s_2,...)}$.
Then, we draw q=2, this involves the sampling of 2 profiled object by $\Gamma\Delta_{2}T_{(s_1,s_2,...)}$.
After that, we get the following profiled tree $A$ :

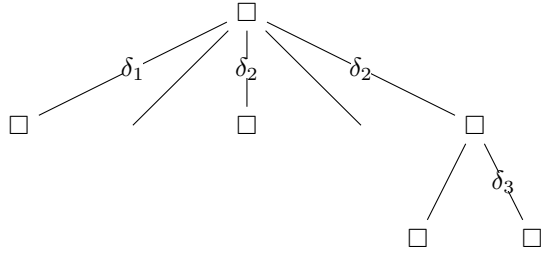
\begin{center}
\begin{tikzpicture}

\node (1) {$\square$} 
  child {  
    node (2) {$\square$} 
    edge from parent 
    node[fill=white,-|,inner sep=-0.8pt] {$\delta_1$}
}
child 
child {
   node (4_1) {$\square$}
    edge from parent
      node[fill=white,-|,inner sep=-0.8pt] {$\delta_2$}
  }
child 
child {
    node (5_1) {$\square$}    
   child {
    node (6_1) {$\square$}
    }
   child{
     node (8_11) {$\square$}
     edge from parent 
     node[fill=white,-|,inner sep=-0.9pt] {$\delta_3$}
   }
   edge from parent 
   node[fill=white,-|,inner sep=-0.8pt] {$\delta_2$}
}
;

\end{tikzpicture}
\end{center}
As the size of $A$ is $14=n$ (So lucky, we are), we just have to return a 14-coloration of $A$ respecting its symetries.
We find :\\
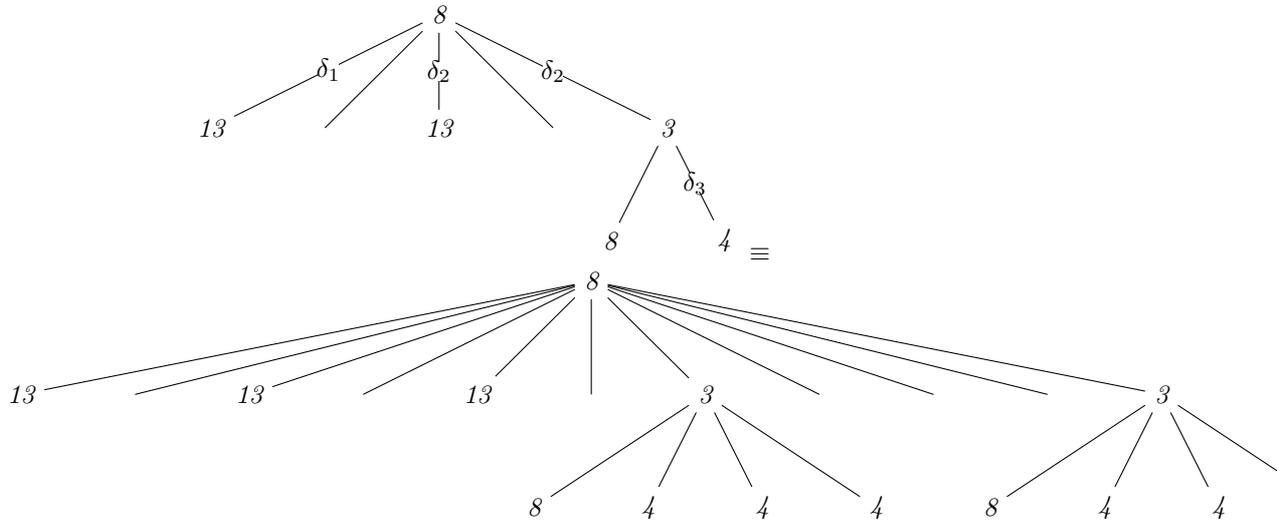
\begin{center}
\begin{tikzpicture}

\node (1) {8} 
  child {  
    node (2) {13} 
    edge from parent 
    node[fill=white,-|,inner sep=-0.8pt] {$\delta_1$}
}
child 
child {
   node (4_1) {13}
    edge from parent
      node[fill=white,-|,inner sep=-0.8pt] {$\delta_2$}
  }
child 
child {
    node (5_1) {3}    
   child {
    node (6_1) {8}
    }
   child{
     node (8_11) {4}
     edge from parent 
     node[fill=white,-|,inner sep=-0.9pt] {$\delta_3$}
   }
   edge from parent 
   node[fill=white,-|,inner sep=-0.8pt] {$\delta_2$}
}
;

\end{tikzpicture}
$\equiv$
\begin{tikzpicture}

\node (1) {8} 
  child {
    node (2) {{13}}    
  }
child 
child {
   node (4_1) {{13}}
  }
child 
child {
    node (4_2) {{13}}
  }
child 
child {
    node (5_1) {3}    
   child {
    node (6_1) {8}
    }
   child{
     node (8_11) {4}
   }
   child{
     node (8_12) {4}
   }
   child{
     node (8_13) {4}
   }
}
child 
child 
child 
child {
    node (5_2) {3}
    child {
      node (6_2) {8}
    }
    child{
      node (8_21) {4}
    }
    child{
      node (8_22) {4}
    }
    child{
      node (8_23) {4}
    }
  };
\end{tikzpicture}
\end{center}

\end{example}
 \begin{theorem}
  The following algorithm is a approximate-size Boltzmann sampler for the size-colored objects in $\mathcal{A}$. Its overall cost is $$O\left(n.((1-\epsilon).\dfrac{f_{\mathcal{A}}(x,n)}{f_{\mathcal{A}}(x,(1-\epsilon)n)}+ (1+\epsilon)^{(1+\epsilon)n}.\dfrac{f_{\mathcal{A}}(x,n)}{f_{\mathcal{A}}(x,(1+\epsilon)n)}\right)$$
 \end{theorem}

\begin{algorithm}[H]
\caption{$\Gamma_{colored}$}

\KwIn{the parameters $n,x_0$}

\KwOut{a colored object.}
Draw a profiled object $a$ with the parameters $s_i=n.x_0^i$\\
Let $\tilde{n}\leftarrow$ size of $a$\\
Let $\tilde{n}_i$ the number of parts of the profile of cardinal $i$.\\
\eIf{$\tilde{n}\notin [(1-\epsilon)n,(1+\epsilon)n]$}{restart $\Gamma_{colored}$}{
\Switch{the value $\tilde{n}$}{
\Case{$\tilde{n}=n$}{
\Return a $n$-coloration of $a$.\\
}
\Case{$\tilde{n}<n$}{
Let $\alpha$ be the minimum in $\{\sum n_i; \sum i n_i=\tilde{n}\mbox{ and }[s_1^{n_1}s_2^{n_2}...]F_{\mathcal{A}}(s_1,s_2,...)\neq 0\}$\\
Draw X following a Bernoulli of parameter $(\dfrac{\tilde{n}}{n})^{-\alpha+\sum \tilde{n}_i} $\\
\eIf{$X=1$}{\Return a $\tilde{n}$-coloration of $a$}{restart $\Gamma_{colored}$}
}
\Case{$\tilde{n}>n$}{
Draw X following a Bernoulli of parameter $(\dfrac{\tilde{n}}{n})^{-\tilde{n}+\sum \tilde{n}_i} $\\
\eIf{$X=1$}{\Return a $\tilde{n}$-coloration of $a$}{restart $\Gamma_{colored}$}
}
}
}
\end{algorithm}

\begin{proof}[Proof of the correctness of $\Gamma_{colored}$]
 Using our sampler of profiled objects with the parameters $s_i=n.x_0^i$, the probability to draw a profiled object of size $\tilde{n}$ having a profile with $\tilde{n}_i$ parts of cardinal $i$ is $\mathbb{P}_n(\mathbf{\tilde{n}})=\dfrac{n^{\sum \tilde{n}_i}.x^{\tilde{n}}.[\mathbf{s}^\mathbf{\tilde{n}}]F_{\mathcal{A}}}{f_{\mathcal{A}}(x,n)}$. We have to add after that a filter in such a way that the probability becomes $\mathbb{P}_{\tilde{n}}(\mathbf{\tilde{n}})=\dfrac{\tilde{n}^{\sum \tilde{n}_i}.x^{\tilde{n}}.[\mathbf{s}^\mathbf{\tilde{n}}]F_{\mathcal{A}}}{f_{\mathcal{A}}(x,\tilde{n})}$. So, we only have to know for which profile the ratio $\dfrac{\mathbb{P}_{\tilde{n}}(\mathbf{\tilde{n}})}{\mathbb{P}_{n}(\mathbf{\tilde{n}})}$ is maximum. If $\tilde{n}<n$ (resp.  $\tilde{n}>n$), this happens when $\sum n_i$ is minimum belongs the $(n_1,n_2,...)$ such that $\sum i n_i=\tilde{n}$ and $[s_1^{n_1}s_2^{n_2}...]F_{\mathcal{A}}(s_1,s_2,...)\neq 0$  (resp. when $\sum n_i$ is maximum belongs the $(n_1,n_2,...)$ such that $\sum i n_i=\tilde{n}$ and $[s_1^{n_1}s_2^{n_2}...]F_{\mathcal{A}}(s_1,s_2,...)\neq 0$). Notice that the maximum is always $\tilde{n}$. The choice of the Bernoulli parameters follows directly from the lemma \ref{lemma}.\\
Under classical conditions explain in \cite{DFLS}, the sampling in approximate-size of a $k$-colored constructible structure can be done in linear time. Now, it suffices to evaluate the cost of the added rejection phase. The average number of loops is $(1-\epsilon).\dfrac{f_{\mathcal{A}}(x,n)}{f_{\mathcal{A}}(x,(1-\epsilon)n)}$ when $\tilde{n}<n$ and $(1+\epsilon)^{(1+\epsilon)n}.\dfrac{f_{\mathcal{A}}(x,n)}{f_{\mathcal{A}}(x,(1+\epsilon)n)}$ otherwise.
\end{proof}

\begin{remark}
 At this stage, the explicit evaluation of the ratio $\dfrac{f_{\mathcal{A}}(x,n)}{f_{\mathcal{A}}(x,(1+\epsilon)n)}$ is out of reach, but it experimentaly seems that it is often better to choose $\tilde{n}>n$. In this case, it is more advisable to modify the line 4 in the previous algorithm by\\ (4') if $\tilde{n}\notin [n,(1+\epsilon)n]$ then restart $\Gamma_{colored}$.
\end{remark}



\section{Conclusion}
We have proved that the profiled objects are powerful tools to generate under Boltzmann model $k$-colored objects and to obtain a generator in approximate size for colored objects. 
Surprisingly, we only have to know the minimum in $\{\sum n_i; \sum i n_i=\tilde{n}\mbox{ and }[s_1^{n_1}s_2^{n_2}...]F_{\mathcal{A}}(s_1,s_2,...)\neq 0\}$ to unbiais the generation of colored objects. We also expect to develop our ideas to even more general coloration types as for instance to impose constraints on the number of atoms with the same colors. Finally, we would like to thank E. Fusy, C. Pivoteau and M. Soria for their precious comments and corrections.

\end{document}